\documentclass[letterpaper, 10 pt, conference]{ieeeconf}  % Comment this line out
\IEEEoverridecommandlockouts                              % This command is only
\usepackage{amsmath}
\usepackage{amssymb}
\usepackage{mathtools}
\usepackage{algpseudocode}
\usepackage[dvipsnames]{xcolor}
\usepackage{setspace}
\usepackage{multirow}
\usepackage{tikz}

\newtheorem{thm}{Theorem}

\newtheorem{defn}{Definition}
\newtheorem{rem}{Remark}
\newtheorem{assum}{Assumption}
\newtheorem{prob}{Problem}
\newtheorem{prop}{Proposition}

\newcommand{\R}{\mathbb{R}}
\newcommand{\Hm}{\mathcal{H}}
\newcommand{\D}{\mathcal{D}}
\newcommand{\C}{\mathcal{C}}

\usetikzlibrary{shapes,arrows,calc,positioning}
\tikzset{
	input/.style={coordinate},
	midpoint/.style={coordinate},
	arrow/.style={draw,-latex,thick},
	line/.style={draw,-,thick},
	sum/.style={draw,circle,node distance=1cm,thick},
	block/.style = {draw, minimum height=1cm, minimum width=2cm, thick, rounded corners=.1cm, align = center},
	blockc/.style = {draw, minimum height=1.75cm, minimum width=4.5cm, thick, rounded corners=.1cm, align = center},
	output/.style = {coordinate,node distance=2cm},
}

\title{\LARGE \bf
An Output Feedback Q-learning Algorithm for Optimal Control of Nonlinear Systems with Koopman Linear Embedding
}

\author{Victor G. Lopez, Malte Heinrich, and Matthias A. Müller% <-this % stops a space
\thanks{This  work was supported by the Deutsche Forschungsgemeinschaft (DFG, German Research Foundation) - 535860958}% <-this % stops a space
\thanks{V. G. Lopez, M. Heinrich, and M. A. Müller are with the Leibniz University Hannover, Institute of Automatic Control, 30167 Hannover, Germany
        {\tt\small \{lopez,mueller\}@irt.uni-hannover.de, malte.heinrich2@stud.uni-hannover.de}}%
}

\begin{document}

\maketitle
\thispagestyle{empty}
\pagestyle{empty}

%%%%%%%%%%%%%%%%%%%%%%%%%%%%%%%%%%%%%%%%%%%%%%%%%%%%%%%%%%%%%%%%%%%%%%%%%%%%%%%%
\begin{abstract}

In the reinforcement learning literature, strong theoretical guarantees have been obtained for algorithms applicable to LTI systems. However, in the nonlinear case only weaker results have been obtained for algorithms that mostly rely on the use of function approximation strategies like, for example, neural networks. In this paper, we study the applicability of a known output-feedback Q-learning algorithm to the class of nonlinear systems that admit a Koopman linear embedding. This algorithm uses only input-output data, and no knowledge of either the system model or the Koopman lifting functions is required. Moreover, no function approximation techniques are used, and the same theoretical guarantees as for LTI systems are preserved. Furthermore, we analyze the performance of the algorithm when the Koopman linear embedding is only an approximation of the real nonlinear system. A simulation example verifies the applicability of this method.

\end{abstract}

%%%%%%%%%%%%%%%%%%%%%%%%%%%%%%%%%%%%%%%%%%%%%%%%%%%%%%%%%%%%%%%%%%%%%%%%%%%%%%%%
\section{INTRODUCTION}

In the past few decades, reinforcement learning (RL) algorithms have been extensively investigated for their use to determine optimal control policies for dynamical systems \cite{SuBa:18,Bert05}. Most of these algorithms do not require any knowledge of the system model, and instead employ measured data to iteratively improve a control policy until it becomes optimal. When the system has linear dynamics and the cost function is quadratic, rigorous theoretical guarantees about stability, convergence and optimality have been obtained \cite{BrYdBa:94,KiumarsiLewJia2017,LopezAlsMue2023,AlsaltiLopMulL4DC2024}. The main reason for this success is the fact that, for linear systems and quadratic costs, only linear policies need to be considered and the solution of the corresponding Bellman equation is known to have a quadratic form.

In contrast, determining optimal controllers for nonlinear systems presents a greater challenge. In general, both the control input and the associated value function are unknown nonlinear functions that are difficult to compute in practice. To address this problem, function approximation techniques are often employed, in particular neural networks (NNs) \cite{Wangetal2024,KiVaMoLe:18,Liuetal:21}. Although NNs are known to have universal function approximation capabilities, their use in RL algorithms complicates their theoretical analysis. Training methods for (very) shallow neural networks have been proposed that allow to obtain some guarantees \cite{KiVaMoLe:18,Liuetal:21,DierksJag2011}, but for such shallow networks the number of nodes required to have accurate approximations may become impractically large. Analyzing RL algorithms that use deeper NNs remains a difficult task. Moreover, even when an adequate NN architecture is used, it is in general unclear which conditions in the measured data guarantee obtaining an accurate function approximation. Persistence of excitation (PE) conditions have been considered \cite{AltamimiLewAbu2008,Liuetal:21} that can be difficult to verify. This contrasts with modern conditions on persistence of excitation for linear systems, where PE verification consists of computing the rank of a matrix \cite{WillemsRapMarDe2005}. 

In \cite{LopezAlsMue2023}, a Q-learning algorithm with strong theoretical guarantees for linear time invariant (LTI) systems was presented. This algorithm exploits the result known as Willems' fundamental lemma \cite{WillemsRapMarDe2005} that provides an elegant condition for persistence of excitation. Using the data collected in this fashion, \cite{LopezAlsMue2023} showed how to reformulate the Bellman equation as a generalized Lyapunov equation that can be solved very efficiently. Then, the resulting algorithm was shown to have a quadratic rate of converge to the optimal control solution. In \cite{AlsaltiLopMulL4DC2024}, this algorithm was extended to the output-feedback case, where the full state of the system is not required to be measured. Instead, \cite{AlsaltiLopMulL4DC2024} exploits the construction of a nonminimal state described in \cite{AlsaltiLopMulTAC2025}, which employs only input-output measurements. The guaranteed properties on convergence and optimality of these algorithms are currently not shared by any RL algorithm for nonlinear systems.

The goal of this paper is to obtain optimal controllers for the class of nonlinear systems that admit a \emph{Koopman linear embedding} (KLE) \cite{ShangCorZhe2024}. Koopman theory is a mathematical tool that allows to lift the dynamics of a nonlinear system into a higher-dimensional space where, in some cases, the dynamics are linear \cite{MauroyMezSus:20}. This can facilitate the design of controllers for nonlinear systems \cite{Strasseretal2026}. Also in the RL context, Koopman theory has been used \cite{Dongeetal2024,SongWanXu2021}. However, most of these methods require the difficult task of determining an adequate set of lifting functions for the given system. In contrast, the result in \cite{ShangCorZhe2024} provided a data-based representation of nonlinear systems with KLE using only input-output data.

In this paper, we first show conditions for the Koopman linear embedding to be stabilizable and observable. Then, we use the output-feedback Q-learning algorithm in \cite{AlsaltiLopMulL4DC2024} to determine an optimal controller for the nonlinear system without knowledge of the model or the lifting functions. We show the conditions under which the theoretical guarantees of this algorithm are preserved: the algorithm yields stabilizing policies at each iteration and converges at a quadratic rate to the optimal control policy. Since, in general, the KLE is not exact, we finally analyze the performance of the algorithm in the presence of approximation errors.

In the remainder of this paper, Section~\ref{secprob} describes the systems with KLE and formulates the main objective of this work. Section~\ref{secmain} presents the Q-learning algorithm that solves the optimal control problem in the case of an exact KLE, and also analyzes the case when there is an approximation error. Section~\ref{secsim} presents a numerical example and Section~\ref{secconc} concludes the paper.

%%%%%%%%%%%%%%%%%%%%%%%%%%%%%%%%%%%%%%%%%%%%%%%%%%%%%%%%%%%%%%%%%%%%%%%%%%%%%%%%
\section{PROBLEM FORMULATION}
\label{secprob}

In this paper, we consider nonlinear systems with a Koopman linear embedding as defined in \cite{ShangCorZhe2024}. In particular, consider a system of the form 
\begin{subequations}
	\begin{align}
		x_{t+1} & = f(x_t,u_t), \\
		y_t & = h(x_t),
	\end{align}
	\label{nlsys}%
\end{subequations}%
where $x_t \in \R^n$, $u_t \in \R^m$, and $y_t \in \R^p$ are the state, input, and output vectors of the system at time $t$, respectively. The functions $f:\R^n \times \R^m \rightarrow \R^n$ and $h:\R^n \rightarrow \R^p$ are unknown continuous functions satisfying $f(0,0)=0$, $h(0)=0$. We let this system satisfy the following assumption.

\begin{assum}
	\label{assstabobs}
	System \eqref{nlsys} is stabilizable in the sense that, for every initial condition $x_0 \in \R^n$, there exist a control sequence $\{ u_t \}_{t=0}^\infty$ such that $x_t \rightarrow 0$ as $t \rightarrow \infty$. Moreover, the system is detectable in the sense that $y_t \rightarrow 0$ as $t \rightarrow \infty$ implies $x_t \rightarrow 0$ as well.  
\end{assum}  

The following definition is taken from \cite{ShangCorZhe2024}. Notice that we consider systems without a direct feedthrough term from input to output.

\begin{defn}[Koopman Linear Embedding]
	\label{defkle}
	System \eqref{nlsys} is said to admit a Koopman linear embedding (KLE) if there exist a scalar $\eta \in \mathbb{N}$, matrices $(A,B,C)$ of appropriate dimensions, and a set of linearly independent functions $\psi_i : \R^n \rightarrow \R$, $i=1,\ldots, \eta$, such that
	\begin{subequations}
		\begin{align}
			\Psi(x_{t+1}) & = A \Psi(x_t) + B u_t, \\
			y_t & = C \Psi(x_t),
		\end{align}
		\label{phikle}%
	\end{subequations}%
	holds for all $t$, where $\Psi(x_t) = [\psi_1(x_t) \quad \cdots \quad \psi_\eta(x_t)]^\top$.
\end{defn}

In the Koopman theory literature, the lifting functions $\psi_i$ are known as \textit{observables}. It is known that general nonlinear systems do not admit an exact KLE \cite{Strasseretal2026,IacobTotSch2024}, and conditions for the existence of the KLE have recently been studied in \cite{Shangetal2026}. However, in many cases the dynamics in \eqref{phikle} may still provide a sufficiently close approximation of the nonlinear dynamics \eqref{nlsys} on some domain of operation. In fact, a variety of controllers for system \eqref{nlsys} designed under the assumption of an exact KLE have been successfully used in practice \cite{Strasseretal2026}.

We are interested in determining a control policy that stabilizes \eqref{nlsys} while minimizing the quadratic cost function
\begin{equation}
	J=\sum_{t=0}^\infty y_t^\top Q y_t + u_t^\top R u_t
	\label{cost}
\end{equation}
with $Q,R \succ 0$. If the system admits a KLE, it may be natural to attempt using a controller of the form $u_t = -K \Psi(x_t)$ to stabilize \eqref{nlsys} by stabilizing \eqref{phikle}. However, it can be very challenging to determine a suitable set of lifting functions $\psi_i$ that satisfy the conditions in Definition~\ref{defkle}. Moreover, in our setting the functions $f$ and $h$ (as well as the matrices $A$, $B$, and $C$) are unknown. Hence, the problem considered in this paper is formulated as follows.

\begin{prob}
	\label{prob1}
	Let the system \eqref{nlsys} admit a KLE as in Definition~\ref{defkle}. Using only input-output data, and without using knowledge of the mathematical models or of the lifting functions $\psi_i$, determine a control policy $u$ that stabilizes \eqref{nlsys} and that minimizes the cost function \eqref{cost}.
\end{prob}

In this paper, we propose the use of a Q-learning algorithm to solve Problem \ref{prob1}. The use of reinforcement learning strategies for control of nonlinear systems has been investigated for many years \cite{KiVaMoLe:18,Liuetal:21}. However, the available theoretical guarantees typically rely on the use of shallow NNs to approximate the solution of the Bellman equation, as well as the resulting nonlinear control policy \cite{KiVaMoLe:18}. For example, 
\begin{equation}
	\hat V(x_t) = W_c \phi_c(x_t), \quad \hat u(x_t) = W_a \phi_a(x_t),
	\label{actorcritic}
\end{equation}
are common expressions for the \emph{critic} approximator and the \emph{actor} approximator, respectively. Here, $\phi_c: \R^n \rightarrow \R^{N_c}$ and $\phi_a : \R^n \rightarrow \R^{N_a}$ are user-defined basis functions, and $W_c \in \R^{1 \times N_c}, \, W_a \in \R^{m \times N_a}$ are the weights to be determined. Training rules have been proposed that lead to bounded function approximation errors \cite{DierksJag2011,KiVaMoLe:18}, hence implying suboptimal controllers. In the following section we show how this issue can be overcome when solving Problem~\ref{prob1} for systems with a Koopman linear embedding.

%%%%%%%%%%%%%%%%%%%%%%%%%%%%%%%%%%%%%%%%%%%%%%%%%%%%%%%%%%%%%%%%%%%%%%%%%%%%%%%%
\section{OPTIMAL CONTROL FOR NONLINEAR SYSTEMS WITH LINEAR KOOPMAN EMBEDDING}
\label{secmain}

Suppose that system \eqref{nlsys} admits a KLE. The expression \eqref{phikle} suggests to consider a linear system of the form
\begin{subequations}
	\begin{align}
		\xi_{t+1} & = A \xi_t + B u_t, \label{xkle}\\
		y_t & = C \xi_t, \label{ykle}
	\end{align}
	\label{kle}%
\end{subequations}%
where $\xi \in \R^{\eta}$ is an auxiliary state variable. Note that, if the initial condition can be written as $\xi_0 = \Psi(x_0)$ for some $x_0$, then $\xi_t = \Psi(x_t)$ for all $t>0$. Throughout this paper, we refer to \eqref{kle} as a KLE of \eqref{nlsys}. In this section, we first analyze the conditions that allow to design optimal controllers for system \eqref{nlsys} by exploiting the linear structure of \eqref{kle}. Then, we present a Q-learning algorithm to solve Problem~\ref{prob1} in the case of an exact KLE and show its theoretical properties. Finally, we consider the case when the KLE is inexact.

\subsection{Conditions for optimal control of systems with KLE}
\label{seccond}

By Definition \ref{defkle}, every input-output trajectory $\{u_t,y_t\}_{t=0}^\infty$ of the nonlinear system \eqref{nlsys} is also an input-output trajectory of the lifted system \eqref{kle}. Hence, it is intuitive that an optimal controller for \eqref{kle} should also be optimal for \eqref{nlsys}. In this subsection we show that this is the case as long as suitable conditions hold. 

As it is discussed in \cite{ShangCorZhe2024}, a given KLE \eqref{kle} may not be controllable and/or observable. However, like any linear state space representation, the realization \eqref{kle} is not unique. It is not difficult to show that, if a system admits a KLE, then there exits an equivalent KLE that is observable.

\begin{prop}
	\label{propobs}
	If system \eqref{nlsys} admits a KLE, then there exists an observable KLE realization.
\end{prop}
\begin{proof}
	Consider an unobservable realization \eqref{kle}. Using Kalman decomposition, there exists a similarity transformation $T_{KD}$ such that
	\begin{equation*}
		T_{KD}^{-1} A T_{KD} = \begin{bmatrix} A_{\bar o} & A_{12} \\ 0 & A_o \end{bmatrix}, \quad C T_{KD} = \begin{bmatrix} 0 & C_o \end{bmatrix},
	\end{equation*}
	and $(A_o,C_o)$ is an observable pair. Hence, the transformation $\bar \xi = [0 \; I] T_{KD}^{-1} \xi$ leads to an observable realization.
\end{proof}

By Proposition~\ref{propobs}, we can assume that the KLE \eqref{kle} is observable without loss of generality. We now show that if \eqref{nlsys} is stabilizable then there is a stabilizable realization \eqref{kle}.

\begin{prop}
	\label{propstab}
	Let system \eqref{nlsys} admit a KLE and let Assumption~\ref{assstabobs} hold. Then, there exists an observable and stabilizable KLE realization.
\end{prop}
\begin{proof}
	An observable realization exists by Proposition~\ref{propobs}. We now show that this realization is also stabilizable. The proof is performed by showing that for any initial state $\xi_0$ there exists an input sequence such that $\xi_t \rightarrow 0$ as $t \rightarrow \infty$. Case 1: the initial state is realizable by the lifting functions as $\xi_0 = \Psi(x_0)$ for some $x_0$. Then, consider the input sequence that, when applied to \eqref{nlsys}, leads to $x_t \rightarrow 0$ (which exists by Assumption~\ref{assstabobs}). Since  the function $h$ in \eqref{nlsys} is continuous, this implies that $y_t \rightarrow 0$. Observability of \eqref{kle} leads to $\xi_t \rightarrow 0$.
	
	Case 2: $\xi_0 \neq \Psi(x_0)$ for any $x_0$. Since the functions $\psi_i$, $i=1,\ldots,\eta$, are linearly independent, there exist vectors $\bar x_0^i$ such that the matrix $[\Psi(\bar x_0^1) \; \cdots \; \Psi(\bar x_0^\eta)]$ is nonsingular. Let the sequence $\bar u^i := \{ \bar u_t^i \}_{t=0}^\infty$ be such that, when applied to \eqref{nlsys} with initial condition $\bar x_0^i$, leads to $\bar x_t^i \rightarrow 0$ as $t\rightarrow \infty$. Moreover, define the vector $\alpha \in \R^\eta$ such that $[\Psi(\bar x_0^1) \; \cdots \; \Psi(\bar x_0^\eta)] \alpha = \xi_0$. Then, from linearity of \eqref{kle} and the arguments in Case 1, it follows that the inputs $u_t = \sum_{i=1}^\eta \alpha_i \bar u_t^i$ lead to $\xi_t \rightarrow 0$.
\end{proof}

In the following theorem, we show that Assumption~\ref{assstabobs} and Propositions~\ref{propobs} and \ref{propstab} together provide the conditions to design stabilizing optimal controllers for system \eqref{nlsys} by exploiting its KLE \eqref{kle}. Notice that conditions for stability and optimality of systems with a Koopman linear embedding were studied in \cite[Section~3]{ShangCorZhe2025} in a different setting. The following result expresses the conditions differently, suitably formulated for the setting considered in this paper. 

\begin{thm}
	\label{thmstabopt}
	Let system \eqref{nlsys} admit a KLE as in Definition~\ref{defkle}, let Assumption~\ref{assstabobs} hold, and consider the cost function \eqref{cost}. There exist lifting functions $\psi_i$, $i=1,\ldots,\eta$, such that the control input $u_t^*= -K^* \Psi(x_t)$ with
	\begin{equation}
		K^*= (R+B^\top P B)^{-1} B^\top PA,
		\label{kopt}
	\end{equation}
	and where $P \succ 0$ satisfies the algebraic Riccati equation
	 \begin{equation}
	 	Q_\xi + A^\top P A - P - A^\top PB (R+B^\top P B)^{-1} B^\top P A = 0,
	 	\label{are}
	 \end{equation}
	 with $Q_\xi = C^\top Q C$, globally asymptotically stabilizes system \eqref{nlsys} and minimizes the cost \eqref{cost}.
\end{thm}

\begin{proof}
	Propositions~\ref{propobs} and \ref{propstab} guarantee the existence of a stabilizable and observable realization \eqref{kle}. Moreover, notice that \eqref{ykle} allows to write the cost \eqref{cost} as
	\begin{equation}
		J=\sum_{t=0}^\infty \xi_t^\top C^\top Q C \xi_t + u_t^\top R u_t.
		\label{costxi}
	\end{equation}
	Hence, \eqref{xkle} and \eqref{costxi} define a standard linear quadratic regulator (LQR) problem that, since \eqref{kle} is stabilizable and observable, is known to have the solution $u_t^*= -K^* \xi_t$ with $K^*$ as in \eqref{kopt}. Since $u^*$ stabilizes \eqref{kle}, $y_t \rightarrow 0$ as $t \rightarrow \infty$, which, by Definition~\ref{defkle}, corresponds also to the output of \eqref{nlsys}. By Assumption~\ref{assstabobs}, this implies $x_t \rightarrow 0$ as $t \rightarrow \infty$. Finally, assume for contradiction that there exists an initial condition $x_0$ and an input policy $\hat u \neq u^*$ such that, when applied to \eqref{nlsys}, yields a better cost $J(x_0,\hat u) < J(x_0,u^*)$. This input would yield the same result when applied to \eqref{kle}, contradicting the optimality of $u^*$.  
\end{proof}

The results obtained in this subsection allow us to design optimal controllers for system \eqref{nlsys} by focusing on the linear expression \eqref{kle}. Although the optimal input in Theorem~\ref{thmstabopt} depends explicitly on the lifting functions $\psi_i$, in the following subsection, we bypass the need for choosing a suitable set of lifting functions by using a Q-learning algorithm that requires only input-output data measured from the system.

\subsection{Output-feedback Q-learning with exact KLE}
\label{secexa}

In this subsection we use the output-feedback Q-learning algorithm in \cite{AlsaltiLopMulL4DC2024} to optimally control \eqref{kle} without information about neither the model nor the lifting functions $\Psi$ (and hence the state $\xi$). This algorithm was shown to be a computationally efficient method with strong theoretical guarantees.

To avoid the need for state information, \cite{AlsaltiLopMulL4DC2024} leveraged the result in \cite{AlsaltiLopMulTAC2025} where an alternative, nonminimal state for an LTI system as \eqref{kle} was constructed using input-output data. Specifically, this state is constructed as
\begin{equation}
	z_t = \begin{bmatrix} u_{[t-\ell,t-1]} \\ \Gamma y_{[t-\ell,t-1]} \end{bmatrix} \in \R^{m \ell+\eta},
	\label{zdef}
\end{equation}
where $\ell$ is the lag (i.e., the observability index) of system \eqref{kle}, $u_{[t-\ell,t-1]}:= [u_{t-\ell}^\top \quad \cdots \quad u_{t-1}^\top]^\top$ and similarly for $y_{[t-\ell,t-1]}$. Moreover, $\Gamma \in \R^{\eta \times p \ell}$ is a constant matrix that can be computed from data as described below. It was shown in \cite{AlsaltiLopMulTAC2025} that, if system \eqref{kle} is controllable, then every state $\xi \in \R^\eta$ can be written as $\xi = Tz$ for some full row rank matrix $T$. This means that $z$ in \eqref{zdef} is a valid state corresponding to a different realization of \eqref{kle}, i.e., there exist matrices $(\bar A, \bar B, \bar C)$ of appropriate dimensions such that
 \begin{subequations}
 	\begin{align}
 		z_{t+1} & = \bar A z_t + \bar B u_t, \label{zsysstate}\\
 		y_t & = \bar C z_t
 	\end{align}
 	\label{zsys}%
 \end{subequations}%
holds. This is also true for stabilizable systems as long as the following assumption is satisfied.

\begin{assum}
	\label{asst}
	There exists a full row rank matrix $T$ such that $\xi = T z$ for any state $\xi$ of \eqref{kle} and with $z$ as in \eqref{zdef}.
\end{assum}

From the proof in \cite[Theorem~4]{AlsaltiLopMulTAC2025}, it can be noted that a sufficient condition for Assumption~\ref{asst} is that the matrix $A$ in \eqref{kle} does not have uncontrollable modes at the origin. Notice that, under Assumption~\ref{asst}, \eqref{zsys} is a valid KLE for system \eqref{nlsys}. Now, a control input of the form $u_t=-Kz_t$ can be used to stabilize \eqref{zsys} while minimizing the cost \eqref{cost}.

The matrix $\Gamma$ in \eqref{zdef} can be obtained using input-output data if persistence of excitation conditions are satisfied. The following definition of PE data is adapted to our setting.

\begin{defn}
	\label{defpe}
	Let the system \eqref{nlsys} admit a KLE \eqref{kle} and let $\D$ be a set containing $\nu$ input-output trajectories of length $\ell+1$ collected from \eqref{nlsys}, i.e., $\D := \left\{ \{ u_t^j,y_t^j \}_{t=0}^{\ell}, j=1,\ldots,\nu \right\}$. The data $\D$ are said to be persistently excited of order $\eta + \ell + 1$ if $\text{rank} \left( \Hm_\D \right) = m(\ell+1)+\eta$, where
	\begin{equation}
		\Hm_\D := \begin{bmatrix}
			u_{[0,\ell]}^1 & u_{[0,\ell]}^2 & \cdots & u_{[0,\ell]}^\nu \\ y_{[0,\ell-1]}^1 & y_{[0,\ell-1]}^2 & \cdots & y_{[0,\ell-1]}^\nu
		\end{bmatrix}.
		\label{hank}
	\end{equation}
\end{defn}

In general, the rank condition in Definition~\ref{defpe} does not imply that the matrix $\Hm_\D$ has full rank. Moreover, it is needed that $\nu \geq m (\ell+1) + \eta$. Note that, as is standard in the data-driven control literature, one can either construct the matrix \eqref{hank} from one long input-output trajectory, or from multiple shorter ones. We also highlight that, different from the PE condition used in \cite{ShangCorZhe2024}, the condition in Definition~\ref{defpe} can be verified from measured data\footnote{If system \eqref{kle} is controllable (and not only stabilizable), then by Willems' lemma \cite{WillemsRapMarDe2005} the conditions in Definition~\ref{defpe} can be enforced by applying a persistently exciting input of order $\eta+\ell+1$ \cite[Lemma~1]{AlsaltiLopMulTAC2025}.}.

\begin{rem}
	In Definition~\ref{defpe}, as well as in the procedure described below, knowledge of the state dimension $\eta$ and the lag $\ell$ of \eqref{kle} is exploited. As it is discussed in \cite[Remark~3]{AlsaltiLopMulTAC2025}, we can replace both $\eta$ and $\ell$ by an upper bound $\tilde \eta \geq \eta \geq \ell$ and still obtain a valid nonminimal state $z_t$.
\end{rem}

For the benefit of the reader, we summarize in Algorithm~1 the steps to compute the matrix $\Gamma$ in \eqref{zdef}, which are described in \cite{AlsaltiLopMulTAC2025}. Algorithm~1 must only be used once to obtain $\Gamma$, and then the nonminimal state $z_t$ can be computed using stored input-output information as in \eqref{zdef} either in an offline fashion (e.g., to run Algorithm~2 below), or in an online fashion (e.g., to apply the controller $u_t=-K z_t$ to the system).

\begin{figure}[h]
	\hrule
	\vspace{0.1cm}
	{\bf Algorithm 1 \cite{AlsaltiLopMulTAC2025}: Construction of $\Gamma$ in \eqref{zdef}}
	{\hrule \small
		\begin{algorithmic}[1]
			\Procedure{}{}
			\State Collect persistently excited data $\D$ of order $\eta+\ell+1$ as in Definition~\ref{defpe} and construct the matrix
			\begin{equation}
				\Hm_\D^- := \begin{bmatrix}
					u_{[0,\ell-1]}^1 & u_{[0,\ell-1]}^2 & \cdots & u_{[0,\ell-1]}^\nu \\ y_{[0,\ell-1]}^1 & y_{[0,\ell-1]}^2 & \cdots & y_{[0,\ell-1]}^\nu
				\end{bmatrix}.
				\label{hankm}
			\end{equation}
			\State Determine a permutation matrix $\Pi \in \R^{p\ell \times p\ell}$ that reorders the rows of $\Hm_\D^-$ such that
			\begin{equation*}
				\Hm_\D^- = \begin{bmatrix}
					I_{m\ell} & 0 \\ 0 & \Pi
				\end{bmatrix} \begin{bmatrix} M \\ N \end{bmatrix},
			\end{equation*}
			where $M \in \R^{m\ell+\eta \times \nu}$ has full row rank, and $N \in \R^{p\ell-\eta \times \nu}$.
			\State Compute the inverse of $\Pi$ and define the matrix $\Gamma \in \R^{\eta \times p\ell}$ such that $\Pi^{-1}=\begin{bmatrix} \Gamma \\ G \end{bmatrix}$ for some $G \in \R^{p\ell-\eta \times p\ell}$.
			\EndProcedure
			\hrule
		\end{algorithmic}
	}
\end{figure}

Algorithm~2 below was proposed in \cite{AlsaltiLopMulL4DC2024} to stabilize system \eqref{zsys} while minimizing the cost \eqref{cost}. There, it was shown that the matrix $Z$ in \eqref{zmat} can always be chosen as nonsingular due to the persistence of excitation condition in Definition~\ref{defpe}. Moreover, as every policy iteration algorithm, an initial stabilizing matrix $K^0$ must be chosen, i.e., $\bar A- \bar BK^0$ must be Schur. An efficient method to determine this matrix from data was proposed in \cite{LopezAlsMue2023}. In the Bellman equation \eqref{bellman} and the control update \eqref{update} we use the definitions
\begin{equation*}
	\bar Q=\begin{bmatrix}
		Q & 0 \\ 0 & R
	\end{bmatrix}, \quad \Theta=\begin{bmatrix}
	\Theta_{zz} & \Theta_{uz}^\top \\ \Theta_{uz} & \Theta_{uu}
	\end{bmatrix},
\end{equation*}
where $Q, R$ are the weighting matrices of the cost function, $\Theta_{uz} \in \R^{m \times m\ell+\eta}$, and $\Theta_{uu} \in \R^{m \times m}$.
 
\begin{figure}[h]
	\hrule
	\vspace{0.1cm}
	{\bf Algorithm 2 \cite{AlsaltiLopMulL4DC2024}: Output-feedback Q-learning}
	{\hrule \small
		\begin{algorithmic}[1]
			\Procedure{}{}
			\State Run Algorithm 1 and store the data $\D$ and the matrix $\Gamma$. Construct the vectors $z_\ell^j$ and $z_{\ell+1}^j$, $j=1,\ldots,\nu$, as in \eqref{zdef}.
			\State Select a set of $\mu = m(\ell+1)+\eta$ indices, $k_j$, $j=1,\ldots,\mu$, from the set $\{ 1,\ldots,\nu \}$ such that the matrix 
			\begin{equation}
				Z=\begin{bmatrix}
					z_\ell^{k_1} & z_\ell^{k_2} & \cdots & z_\ell^{k_\mu} \\ u_\ell^{k_1} & u_\ell^{k_2} & \cdots & u_\ell^{k_\mu}
				\end{bmatrix} \in \R^{\mu \times \mu}
				\label{zmat}
			\end{equation}
			is nonsingular. Moreover, define
			\begin{equation}
				W=\begin{bmatrix}
					y_\ell^{k_1} & y_\ell^{k_2} & \cdots & y_\ell^{k_\mu} \\ u_\ell^{k_1} & u_\ell^{k_2} & \cdots & u_\ell^{k_\mu}
				\end{bmatrix} \in \R^{m+p \times \mu}.
			\end{equation}
			\State Set $i=0$ and determine an initial stabilizing matrix $K^0$.
			\State Using $K^i$, construct the matrix
			\begin{equation}
				\Sigma_i=\begin{bmatrix}
					z_{\ell+1}^{k_1} & z_{\ell+1}^{k_2} & \cdots & z_{\ell+1}^{k_\mu} \\ -K^iz_{\ell+1}^{k_1} & -K^iz_{\ell+1}^{k_2} & \cdots & -K^iz_{\ell+1}^{k_\mu}
				\end{bmatrix} \in \R^{\mu \times \mu}.
			\end{equation}
			 and solve for $\Theta^{i+1}$ from the equation 
			\begin{equation}
				Z^\top \Theta^{i+1} Z = W^\top \bar Q W + \Sigma_i^\top \Theta^{i+1} \Sigma_i.
				\label{bellman}
			\end{equation}
			\State Update the input policy as
			\begin{equation}
				K^{i+1}=\left( \Theta_{uu}^{i+1} \right)^{-1} \Theta_{uz}^{i+1}.
				\label{update}
			\end{equation}
			\State Let $i=i+1$ and go to Step 5. On convergence, stop. 
			\EndProcedure
			\hrule
		\end{algorithmic}
	}
\end{figure}

The following theorem shows that a control input of the form $u_t^* = -\bar K^* z_t$ is optimal for system \eqref{nlsys}, and that Algorithm~2 converges to the matrix $\bar K^*$.

\begin{thm}
	\label{thmsol}
	Let system \eqref{nlsys} admit a KLE as in Definition~\ref{defkle}, consider the cost function \eqref{cost}, and let Assumptions~\ref{assstabobs} and \ref{asst} hold. Then, a control input of the form $u_t^* = -\bar K^* z_t$, with $z_t$ as in \eqref{zdef}, stabilizes \eqref{nlsys} and minimizes the cost \eqref{cost}. Moreover, at every iteration of Algorithm~2, the Bellman equation \eqref{bellman} has a unique solution $\Theta^{i+1}$, the input $u_t^i = -K^i z_t$ stabilizes \eqref{nlsys}, and the algorithm converges quadratically as $\lim_{i\rightarrow \infty} K^i = \bar K^*$.
\end{thm}
\begin{proof}
	By Propositions~\ref{propobs} and \ref{propstab}, \eqref{kle} is a stabilizable and observable KLE for \eqref{nlsys} without loss of generality. By Theorem~\ref{thmstabopt}, the input $u_t^* = -K^* \xi_t$, with $K^*$ as in \eqref{kopt}-\eqref{are}, is stabilizing and optimal. By Assumption~\ref{asst}, we can write this optimal input as $u_t^* = -K^* T z_t = -\bar K^* z_t$, with $\bar K^* := K^* T$. It was shown in \cite{AlsaltiLopMulL4DC2024} that $K^i$ from Algorithm~2 quadratically converges to $\bar K^*$. Finally, it was also shown in \cite{AlsaltiLopMulL4DC2024} that the inputs $u_t^i = -K^i z_t$ obtained at each iteration stabilize \eqref{kle} and, hence, it can be shown as in the proof of Theorem~\ref{thmstabopt} that they also stabilize \eqref{nlsys}.
\end{proof}

Theorem~\ref{thmsol} provides our solution to Problem~\ref{prob1}. The input $u_t^* = -\bar K^* z_t$ optimally stabilizes system \eqref{nlsys} using only input-output information (cf. \eqref{zdef}). We remark that the input $u_t^* = -\bar K^* z_t$ is not linear in the original state $x$. Note from \eqref{zdef} that this is a dynamic output feedback controller that implicitly accounts for the nonlinearities of the system.

\subsection{On an inexact linear embedding and noisy data}
\label{secinex}

It is known that general nonlinear systems do not admit a Koopman linear embedding \cite{Strasseretal2026,IacobTotSch2024}. However, a KLE of sufficiently high order may provide a close enough approximation of the system dynamics for a KLE-based control design to be effective \cite{Strasseretal2026}. A thorough study of the effects of approximation errors in our proposed method is a subject of future work, but in this subsection we present initial results regarding the inherent robustness properties of Algorithm~2.

First, we analyze the case when the lifted states do not propagate exactly in a linear fashion, but the output does depend linearly on the lifted state. That is, instead of the system \eqref{zdef}, consider the system
\begin{subequations}
	\begin{align}
		z_{t+1} & = \bar A z_t + \bar B u_t + \varepsilon_t, \\
		y_t & = \bar C z_t,
	\end{align}
	\label{ikle1}%
\end{subequations}%
where  $\varepsilon_t \in \R^{m\ell+\eta}$ corresponds to the modeling error in the state dynamics at each time $t$. We consider this case because it provides interesting insights about the convergence of Algorithm~2. A more general case is discussed later.

The procedure that we follow now is an extension for our setting of the analysis performed in \cite{LopezAlsMue2023}. First, from the data collected from \eqref{ikle1}, we can write
\begin{equation*}
	W = \begin{bmatrix} \bar C z_\ell^{k_1} & \cdots & \bar C z_\ell^{k_\mu} \\ u_\ell^{k_1} & \cdots & u_\ell^{k_\mu} \end{bmatrix} = \C Z,
\end{equation*}
\begin{multline*}
	\Sigma_i = \begin{bmatrix} \bar A & \bar B \\ -K^i \bar A & -K^i \bar B \end{bmatrix} \begin{bmatrix} z_\ell^{k_1} & \cdots & z_\ell^{k_\mu} \\ u_\ell^{k_1} & \cdots & u_\ell^{k_\mu} \end{bmatrix} \\
	+ \begin{bmatrix} I \\ -K^i \end{bmatrix} \begin{bmatrix} \varepsilon_\ell^{k_1} & \cdots & \varepsilon_\ell^{k_\mu} \end{bmatrix} = \Phi_i Z + \mathcal{K}^i E.
\end{multline*}
where 
\begin{equation*}
	\Phi_i = \begin{bmatrix} \bar A & \bar B \\ -K^i \bar A & -K^i \bar B \end{bmatrix}, \quad \mathcal{C} = \begin{bmatrix} \bar C & 0 \\ 0 & I \end{bmatrix}, \quad \mathcal{K}^i=\begin{bmatrix} I \\ -K^i \end{bmatrix},
\end{equation*}
and $E = [\varepsilon_\ell^{k_1} \, \cdots \, \varepsilon_\ell^{k_\mu}]$. Pre- and post-multiplying \eqref{bellman} by $(Z^{-1})^\top$ and $Z^{-1}$, respectively, we get
\begin{equation*}
	\Theta^{i+1} = \C^\top \bar Q \C  + \left( \Phi_i + \mathcal{K}^i E Z^{-1} \right)^\top \Theta^{i+1} \left( \Phi_i + \mathcal{K}^i E Z^{-1} \right).
\end{equation*}
In the absence of modeling errors, this expression becomes
\begin{equation*}
	\Theta^{i+1} = \C^\top \bar Q \C + \Phi_i^\top \Theta^{i+1} \Phi_i,
\end{equation*}
which is a Lyapunov equation that has a unique solution $\Theta^i \succ 0$ if and only if $\Phi_i$ is Schur, which happens if and only if $\bar A - \bar B K^i$ is Schur \cite{LopezAlsMue2023}.

Now, since $\Phi^i=\mathcal{K}^i [\bar A \; \bar B ]$, then $\Phi_i +\mathcal{K}^i E Z^{-1} = \mathcal{K}^i ([\bar A \; \bar B ] + E Z^{-1}) = \mathcal{K}^i [\hat A \; \hat B ]=: \hat \Phi^i$, where $[\hat A \; \hat B ] := [\bar A \; \bar B ] + E Z^{-1}$. From \eqref{ikle1} notice that, in fact, $[\hat A \; \hat B ]$ is the solution of the least-squares problem
\begin{equation*}
	\min_{\hat A, \hat B} \left\| \begin{bmatrix}  z_{\ell+1}^{k_1} & \cdots & z_{\ell+1}^{k_\mu} \end{bmatrix} - \begin{bmatrix} \hat A & \hat B \end{bmatrix} Z \right\|^2.
\end{equation*}
Therefore, using the data from \eqref{ikle1}, Algorithm~2 solves the optimal control problem for the least-squares linear approximation of \eqref{ikle1}, but without computing such matrices. Now, the algorithm converges if $K^0$ is stabilizing for the system defined by $(\hat A, \; \hat B)$.

We finally comment on the case when the input-output trajectories of \eqref{nlsys} are trajectories of the system
\begin{subequations}
	\begin{align}
		z_{t+1} & = \bar A z_t + \bar B u_t + \varepsilon_t,\\
		y_t & = \bar C z_t + v_t,
	\end{align}
	\label{ikle2}%
\end{subequations}%
where  $\varepsilon_t \in \R^{\eta}$ and $v_t \in \R^p$ are the modeling errors, and they can account for noise as well. The first issue in this case is that the state $z$ in \eqref{zdef} must be computed from ``noisy" outputs. This often makes the matrix \eqref{hankm} have full rank, preventing an adequate computation of $\Gamma$. If the mismatch terms $\varepsilon$ and $v$ have sufficiently small magnitudes, this problem can be solved by following a singular value decomposition approach, as suggested in \cite{AlsaltiLopMulTAC2025}.

Moreover, the additional term $v_t$ in \eqref{ikle2} implies that the matrix $W$ now has the form
\begin{equation*}
	W = \begin{bmatrix} \bar C z_\ell^{k_1} + v_\ell^{k_1} & \cdots & \bar C z_\ell^{k_\mu} + v_\ell^{k_\mu} \\ u_\ell^{k_1} & \cdots & u_\ell^{k_\mu} \end{bmatrix} = \C Z + V,
\end{equation*}
where the columns of $V$ are $[(v_\ell^{k_j})^\top \; 0]^\top$, $j=1,\ldots,\mu$. Hence, as in the previous case, from \eqref{bellman} we get
\begin{equation*}
	\Theta^{i+1} = \hat Q  + \left( \Phi_i + \mathcal{K}^i E Z^{-1} \right)^\top \Theta^{i+1} \left( \Phi_i + \mathcal{K}^i E Z^{-1} \right),
\end{equation*}
with $\hat Q = \left( \C + V Z^{-1} \right)^\top \bar Q \left( \C + V Z^{-1} \right)$. Since $\bar Q$ contains the cost's weights, this implies solving an optimization problem for a ``disturbed cost", as well as considering the least-squares system given by $(\hat A, \hat B)$ as in the previous case. It is a matter of future research to investigate how this affects the controller performance.

In the following section, we show how the proposed Q-learning controller compares with a traditional RL algorithm that uses neural network approximations.

%\addtolength{\textheight}{-2.5cm}   %

%%%%%%%%%%%%%%%%%%%%%%%%%%%%%%%%%%%%%%%%%%%%%%%%%%%%%%%%%%%%%%%%%%%%%%%%%%%%%%%%
\section{SIMULATION EXAMPLE}
\label{secsim}

Consider a system with $x_t \in \R^3$ and $u_t \in \R^2$ given by
\begin{equation*}
	x_{t+1} = \begin{bmatrix}
		0.7 x_{1,t} \\ 0.9 x_{2,t} + x_{1,t}^3 - x_{1,t}^2 \\ 0.8 x_{3,t} -x_{2,t} + x_{1,t}^5
	\end{bmatrix} + \begin{bmatrix}
	0 & 0 \\ 1 & 0 \\ 0 & 1
	\end{bmatrix} u_t.
\end{equation*}
The output is $y_t=x_t \in \R^3$, i.e., state measurements are available. This is not necessary for the application of our method, but it allows to compare with the RL algorithm in $\cite{DierksJag2011}$ that requires state information. This system admits a KLE, for example by taking $\Psi(x_t)=[x_{1,t} \;\; x_{2,t} \;\; x_{3,t} \;\; x_{1,t}^2 \;\; x_{1,t}^3 \;\; x_{1,t}^5]^\top \in \R^6$. Information about these functions is not used in Algorithm~2. The cost function is chosen as \eqref{cost} with $Q=I_3$ and $R=I_2$. 

Data is collected by applying a uniformly distributed random input to the system, $u \in \mathcal{U}(-1,1)^m$. The rank condition in Definition~2 was verified. Since the system is stable, $K^0 = 0$ is an initial stabilizing policy. We compare the performance of Algorithm~2 with the RL algorithm in \cite{DierksJag2011}, where NNs as in \eqref{actorcritic} are used for the actor-critic approximators. Since the user has no knowledge of the system model, we test two different sets of basis functions choices. First, the basis functions in \eqref{actorcritic} are chosen as $\phi_c(x_t) = [(x_t \otimes x_t)^\top \quad x_{1,t}^4 \quad x_{1,t}^6]^\top$ and $\phi_a(x_t) = [x_t^\top \quad (x_t \otimes x_t)^\top]^\top$, where $\otimes$ is the Kronecker product (call this example RL-NN1). This choice of functions can only inexactly approximate the optimal controller. Then (example RL-NN2), we repeat the experiment but using the basis functions $\phi_c(x_t) = \Psi(x_t) \otimes \Psi(x_t)$ and $\phi_a(x_t) = \Psi(x_t)$, where $\Psi$ was defined above. This selection of basis functions is rich enough to yield the optimal controller, if the optimal weights $W_c, \, W_a$ are obtained. 

After running each of the algorithms, we apply the resulting controllers to the system for the same initial conditions, and compute the corresponding costs. Algorithm~2 is run for only ten iterations. To achieve an acceptable performance with the algorithm in \cite{DierksJag2011}, we execute it for 10000 iterations. The resulting controllers are applied to the system 100 times for random initial conditions. The relative error between the resulting cost and the optimal cost (obtained from model knowledge) is computed and averaged. The results, together with the computation time, are displayed in Table~\ref{tab1}. This shows the superior performance of our algorithm. Note that, although the example RL-NN2 had the potential to reach optimality, it failed to do so because the training method in \cite{DierksJag2011} does not guarantee convergence of the NN weights to their ideal values. The proposed Algorithm~2 achieves optimal performance with a much lower computation time, and without the need to know the vector of functions $\Psi(x)$.  

\begin{table}
	\centering
	\begin{tabular}{|c|c|c|c|}
		\hline 
		& \# iter & avg. cost error  & avg. time ($\mathrm{s}$) \\
		\hline 
		QL
		& $10$ & $1.0127\times 10^{-16}$ & $2.012 \times 10^{-3}$\\    
		\hline
		RL-NN1
		& $10^4$  & $4.3547\times 10^{-2}$ & $7.676 \times 10^{0}$\\    
		\hline
		RL-NN2
		& $10^4$ & $1.0167\times 10^{-3}$ & $9.089 \times 10^{0}$\\    
		\hline
	\end{tabular}
	\vspace{4pt}
	\caption{Comparison between Algorithm 2 (QL) and the RL algorithm in \cite{DierksJag2011} for two different NNs (RL-NN1 and RL-NN2).}
	\label{tab1}
\end{table}%

%%%%%%%%%%%%%%%%%%%%%%%%%%%%%%%%%%%%%%%%%%%%%%%%%%%%%%%%%%%%%%%%%%%%%%%%%%%%%%%%
\section{CONCLUSION}
\label{secconc}

In this paper, we showed that a Q-learning algorithm with strong stability and convergence properties can be used to determine optimal controllers for the class of nonlinear systems that admit a Koopman linear embedding. Conditions for the existence of an exact KLE are given in \cite{Shangetal2026}, and they include, for example, certain classes of polynomial systems (cf. the system in Section~\ref{secsim}). In this case, the optimal controller can be obtained without any knowledge of the mathematical model, the lifting functions, or the system state. From the analysis in Section~\ref{secinex} it can be noted that, when the KLE is inexact, as the magnitude of the approximation error goes to zero, the optimal controller is recovered. Future research includes designing algorithms with robustness properties that can better overcome the KLE approximation errors that are encountered in most applications.

%\addtolength{\textheight}{-3cm}   % This command serves to balance the column lengths
% on the last page of the document manually. It shortens
% the textheight of the last page by a suitable amount.
% This command does not take effect until the next page
% so it should come on the page before the last. Make
% sure that you do not shorten the textheight too much.

\bibliographystyle{IEEEtran}
\bibliography{IEEEabrv,klerl_refs}

\end{document}